\documentclass[aps,twocolumn,showpacs,amsmath,amssymb]{revtex4}
\usepackage{amsthm}
\usepackage{latexsym}
\usepackage{amsfonts}
\usepackage{amssymb}
\usepackage{bbm,dsfont}
\usepackage{color}
\usepackage{graphicx}

\newtheorem{proposition}{Proposition}

\newtheorem{lemma}{Lemma}
\newtheorem{corollary}{Corollary}

\theoremstyle{definition}

\newtheorem{remark}{Remark}
\newtheorem{example}{Example}
\newtheorem{definition}{Definition}

\newcommand{\e}{{\rm e}}

\definecolor{darkgreen}{rgb}{0,0.4,0.3}
\newcommand{\red}[1]{#1}

\newcommand{\R}{\mathbb{R}} %real
\newcommand{\C}{\mathbb{C}} %complex
\newcommand{\T}{\mathbb{T}} %torus
\newcommand{\real}{\mathbb R} %real
\newcommand{\complex}{\mathbb C}
\newcommand{\torus}{\mathbb T}
\newcommand{\Z}{\mathbb Z} %integer
\newcommand{\frecc}{\to}

\newcommand{\hi}{\mathcal{H}} %Hilbert space
\newcommand{\hh}{\mathcal{H}} %Hilbert space
\newcommand{\vv}{\mathcal{V}} %Hilbert space
\newcommand{\kk}{\mathcal{K}} %Hilbert space
\renewcommand{\aa}{\mathcal{A}} %von Neumann algebra
\newcommand{\hik}{\mathcal{K}} %Hilbert space
\newcommand{\lh}{\mathcal{L(H)}} %bounded linear operators
\newcommand{\elle}[1]{\mathcal{L} \left( #1 \right)} %bounded linear operators on #1
\newcommand{\spanno}[1]{{\rm span}\,\left\{ #1 \right\}}
\newcommand{\ip}[2]{\left\langle\,#1\,|\,#2\,\right\rangle} %inner product
\newcommand{\kb}[2]{|#1\rangle\langle#2|} %ketbra
\newcommand{\no}[1]{\left\|#1\right\|} %norm
\newcommand{\tr}[1]{{\rm tr}\left[#1\right]} %trace
\newcommand{\id}{\mathbbm{1}} %identity operator
\newcommand{\rank}[1]{{\rm rank}\,#1} %rank

\newcommand{\mat}[1]{\left( \begin{array}{cc} #1 \end{array} \right)} %2*2 matrix
\newcommand{\mattt}[1]{\left( \begin{array}{ccc} #1 \end{array} \right)} %3*3 matrix

\newcommand{\Mo}{\mathsf{M}}%generic observable
\renewcommand{\S}{\mathcal{S}} 

\begin{document}%\setlength{\arraycolsep}{2pt}

\title[Minimal covariant observables identifying all pure states]{Minimal covariant observables identifying all pure states}

\author{Claudio Carmeli}
\email{claudio.carmeli@gmail.com}
\affiliation{D.I.M.E., Universit\`a di Genova, Via Cadorna 2, I-17100 Savona, Italy, and I.N.F.N., Sezione di Genova, Via Dodecaneso, 33, I-16146 Genova, Italy}

\author{Teiko Heinosaari}
\email{teiko.heinosaari@utu.fi}
\affiliation{Turku Centre for Quantum Physics, Department of Physics and Astronomy, University of Turku, Finland}

\author{Alessandro Toigo}
\email{alessandro.toigo@polimi.it}
\affiliation{Dipartimento di Matematica, Politecnico di Milano, Piazza Leonardo da Vinci 32, I-20133 Milano, Italy, and I.N.F.N., Sezione di Milano, Via Celoria 16, I-20133 Milano, Italy}


\begin{abstract}
It has been recently shown that an observable that identifies all pure states of a $d$-dimensional quantum system has minimally $4d-4$ outcomes or slightly less (the exact number depending on the dimension $d$) \cite{HeMaWo11}.
However, no simple construction of this type of observable with minimal number of outcomes is known.
In this work we investigate the possibility to have a covariant observable that identifies all pure states and has minimal number of outcomes for this purpose.
It is shown that the existence of these kind of observables depends on the dimension of the Hilbert space.
The fact that these kind of observables fail to exist in some dimensions indicates that the dual pair of observables -- pure states lacks the symmetry that the dual pair of observables -- states has.
\end{abstract}

\maketitle

%%%%%%%%%%%%%%%%%%%%%%
\section{Introduction}\label{sec:intro}
%%%%%%%%%%%%%%%%%%%%%%

An observable is called \emph{informationally complete} if it identifies all states \cite{Prugovecki77}.
An informationally complete observable acting on a finite dimensional Hilbert space of dimension $d$ is called \emph{minimal} if it has as small number of outcomes as possible, and this smallest number is known to be $d^2$ \cite{CaFuSc02JMP}.

Often, it is enough to concentrate on pure states.
For instance, we may have a prior knowledge that the otherwise unknown state is pure and we want to identify it.
A measurement that identifies all pure states may not be able to identify all states.
In particular, this kind of measurement can have have less outcomes than $d^2$ whenever $d\geq 3$.
The minimal number of outcomes is roughly $4d-4$ and always less or equal to this \cite{HeMaWo11}.
In spite of the quite precise knowledge about the minimal number of outcomes, a simple or instructive construction of an observable that identifies all pure states and has minimal number of outcomes for this purpose is not known.

The present paper concentrates on minimal \emph{covariant} observables identifying all pure states.
To explain the problem, we start by recalling the standard construction of a minimal covariant observable identifying all states.
First, fix an orthonormal basis $\{h_0,\ldots,h_{d-1}\}$ for a $d$-dimensional Hilbert space $\hi_d$ and let $U$ be the shift operator in this basis, i.e., $Uh_\ell = h_{\ell+1}$, with summation modulo $d$.
Then $j\mapsto U^j$ is a unitary representation of the cyclic group $\Z_d$.
We then define $V$ to be the unitary operator $Vh_\ell = e^{2\pi i \ell /d} h_\ell$, hence giving another unitary representation $k\mapsto V^k$ of the group $\Z_d$. 
The operators $U$ and $V$ commute up to a scalar factor, $VU=e^{2\pi i /d} UV$, thus the combined map $W(j,k):=U^jV^k$ is a projective unitary representation of the product group $\Z_d\times\Z_d$.
By fixing a positive operator $M$ satisfying $\tr{M}=1/d$ and setting
\begin{equation}
\Mo(j,k):=W(j,k) M W(j,k)^\ast \, ,  \qquad j,k\in\Z_d
\end{equation}
we obtain an observable $\Mo$ with $d^2$ outcomes.
This construction also guarantees that $\Mo$ is \emph{covariant}, i.e., 
\begin{equation}\label{eq:cpso}
W(j',k') \Mo(j,k) W(j',k')^\ast = \Mo(j+j',k+k') 
\end{equation}
for all $(j,k), (j',k')\in\Z_d\times\Z_d$.
In order to obtain informationally complete observable one has to choose $M$ such that $\tr{M W(j,k)}\neq 0$ for all $j,k\in\Z_d$; see e.g. \cite{DaPeSa04}, \cite{CaHeTo12}.

In this paper we address the following question:
Is it possible to construct a covariant observable that identifies all pure states and has the minimal number of outcomes for this purpose?
For dimension $2$ the answer is positive since a qubit observable identifies all pure states only if it identifies all states \cite{HeMaWo11}, thus the previous construction gives an observable with the desired properties.
One could presume that a similar construction should be possible in all dimensions, perhaps using a different symmetry group.
 However, we prove that the answer to the existence question depends on the dimension of the Hilbert space.
More precisely, we construct a class of observables with the required properties in dimension $3$ using the group of unit quaternions, but we show that they do not exist in dimension $7$ (and certain other dimensions).

An interesting additional fact is that while the previously sketched general construction uses an irreducible projective representation of an abelian group, the analogous construction for a pure-state informationally complete observable in dimension $3$ is forced to use a reducible unitary representation of a non-abelian group.

Our investigation is organized as follows.
Sections \ref{sec:pic} and \ref{sec:covariant} give some essential background information on pure-state informationally complete observables and covariant observables, respectively.
In Section \ref{sec:existence} we derive some general conditions for the existence of covariant pure-state informationally complete observables.
Then, in Section \ref{sec:minimal} we add the requirement of minimal number of outcomes and prove the main results.
Finally, in Section \ref{sec:conclusions} we outline our conclusions and sketch some further developments and generalizations.

%%%%%%%%%%%%%%%%%%%%%%
\section{Pure-state informationally complete observables}\label{sec:pic}
%%%%%%%%%%%%%%%%%%%%%%

Throughout this paper, by a \emph{Hilbert space} we will always mean a finite dimensional complex Hilbert space. If $\hi$ is a Hilbert space, we denote by $\lh$ the vector space of all linear operators on $\hi$.
It is a Hilbert space itself when equipped with the Hilbert-Schmidt inner product, 
\begin{equation*}
\ip{L_1}{L_2}_{HS}=\tr{L_1^\ast L_2} \, .
\end{equation*}
A quantum state is described by a positive operator $\varrho\in\lh$ satisfying $\tr{\varrho}=1$; it is a pure state if $\varrho$ is a $1$-dimensional projection.

Quantum observables are described and identified with positive operator valued measures (POVMs).
A POVM with finite number of outcomes is a mapping $\Mo$ from a finite set $\Omega$ into positive operators on $\hi$, which is required to satisfy the normalization $\sum_x \Mo(x)=\id$; here $\id$ is the identity operator on $\hi$.
The probability of obtaining an outcome $x$ when the system is in a state $\varrho$ is $\tr{\varrho\Mo(x)}$.

We will be interested on the following two properties.

\begin{definition}
An observable $\Mo$ is called
\begin{itemize}
\item[(a)]  \emph{informationally complete (IC)} if for any two different states $\varrho_1,\varrho_2$, there is at least one outcome $x\in \Omega$ such that $\tr{\varrho_1 \Mo(x)} \neq \tr{\varrho_2 \Mo(x)}$.
\item[(b)]  \emph{pure-state informationally complete (PIC)} if for any two different pure states $\varrho_1,\varrho_2$, there is at least one outcome $x\in \Omega$ such that $\tr{\varrho_1 \Mo(x)} \neq \tr{\varrho_2 \Mo(x)}$.
\end{itemize}
\end{definition}

To be able to write down mathematical criteria for the above properties, we introduce some additional notation.
For each observable $\Mo$, we denote by $\S_\Mo$ the linear span of the range of $\Mo$, i.e., 
\begin{equation*}
\S_\Mo := \spanno{ \Mo(x) : x\in \Omega} = \{ \Sigma_x c_x \Mo(x) : c_x\in \complex \} \, .
\end{equation*}
It is well-known that an observable $\Mo$ is informationally complete if and only if $\S_\Mo=\lh$ \cite{Busch91}.

For any subset $\S\subseteq\lh$, we denote by $\S^\perp$ the orthogonal complement of $\S$ in the Hilbert-Schmidt inner product.
In particular, for each observable $\Mo$ we have
\begin{align*}
\S_\Mo^\perp  = & \{ T\in\lh: \tr{ T^\ast S}=0 \quad\forall S\in \S_\Mo \} \\
 = & \{ T\in\lh: \tr{T\Mo(x)}=0 \quad\forall x\in \Omega \} \, .
\end{align*}
Clearly, $\S_\Mo^\perp$ is a linear subspace of $\lh$ and $\tr{T}=0$ for every $T\in \S_\Mo^\perp$ because 
\begin{equation*}
\tr{T}=\tr{T \id} = \sum_x \tr{T\Mo(x)} = 0 \, .
\end{equation*}

If $\hi$ is $d$-dimensional, then $\lh$ has dimension $d^2$.
Since $\lh$ can be written as $\lh=\S_\Mo \oplus \S_\Mo^\perp$, we obtain
\begin{equation*}
\dim (\S_\Mo) + \dim (\S_\Mo^\perp) = d^2 \, .
\end{equation*}
Notice that if $\Mo$ is injective and the set $\{ \Mo(x) : x\in \Omega\}$ is linearly independent, then $\dim (\S_\Mo)$ is just the number of outcomes of $\Mo$.

We recall the following result from \cite{HeMaWo11}.
\begin{proposition}\label{prop:pic}
An observable $\Mo$ is pure-state informationally complete if and only if every nonzero selfadjoint operator in $\S_\Mo^\perp$ has rank $3$ or more. 
\end{proposition}

If $\S_\Mo^\perp=\{0\}$ (equivalently $\S_\Mo=\lh$), then the condition in Proposition \ref{prop:pic} is satisfied.
But then $\Mo$ is actually informationally complete with respect to all states and must have at least $d^2$ outcomes.

To obtain an observable $\Mo$ that is pure-state informationally complete and has as few outcomes as possible, we obviously need to choose $\S_\Mo$ to have as small dimension as possible while assuring that the condition stated in Proposition \ref{prop:pic} holds. 
The first basic question then arises:
\begin{quote}
If a linear subspace $\S\subseteq\lh$ is given, is it possible to find an observable $\Mo$ such that $\S=\S_\Mo$?
\end{quote}
An easy construction \cite{HeMaWo11} shows that the question has an affirmative answer if and only if
\begin{equation}\label{eq:operator-system}
\id\in \S \qquad \textrm{and} \qquad \S^\ast = \S \, , 
\end{equation}
where the last condition means that $L^\ast\in \S$ whenever $L\in \S$.
Moreover, in this case there exists $\Mo$ with exactly $\dim(\S)$ outcomes such that $\S = \S_\Mo$, but not with less.

A more complicated task is to deduce the smallest possible dimension of a subspace $\S\subseteq\lh$ satisfying \eqref{eq:operator-system} and the criterion of Proposition \ref{prop:pic}, hence giving the minimal number of outcomes for a PIC observable.
It was proved in \cite{HeMaWo11} that the minimal number of outcomes for a PIC observable in $d$-dimensional Hilbert space is $4d-4-\delta(d)$, where $0\leq \delta(d) \leq 2 \log_2(d)$.
The minimal numbers for the dimensions $2-15$ are listed in Table \ref{tab:min}.

\begin{table}
\caption{The minimal number of outcomes for a PIC observable in the dimensions $2-15$.
The exact value is not known in every dimension, but the uncertainty is at most $2 \log_2(d)$.}
\label{tab:min}
\begin{center}
\begin{tabular}{ | c | c |}
\hline
dimension & min $\#$ of outcomes \\
\hline
\hline
2  & 4 \\ 
3 & 8 \\ 
4 & 10 \\ 
5 & 16 \\ 
6 & 18 \\ 
7 & 23 \\ 
8 & 24 or 25 \\
9 & 32 \\
10 & 34 \\
11 & 39 \\
12 & 40 or 41 \\
13 & 47 \\  
14 & 48 or 49 \\
15 & 54 \\
\hline
\end{tabular}
\end{center}
\end{table}

%%%%%%%%%%%%%%
\section{Covariant observables}\label{sec:covariant}
%%%%%%%%%%%%%%

%%%%%%%%%%%%%%%%%%%
\subsection{Projective unitary representations}
%%%%%%%%%%%%%%%%%%%

The usual complication in any implementation of symmetry in a quantum system is that we cannot apriori restrict to unitary representations; instead, we have to deal with more general objects of projective unitary representations. 
In this subsection we recall some basic concepts and results related to projective unitary representations of finite groups. 
For more details, we refer to \cite{Isaacs}.

Let $G$ be a finite group.
(In the rest of this paper, all groups are assumed to be finite even if this is not constantly repeated.)
A \emph{projective unitary representation} of $G$ is a mapping $g\mapsto U(g)$ from $G$ into the set of unitary operators on $\hi$ such that $U(e)=\id$ and
\begin{equation}\label{eq:projective-rep}
U(gh)=\omega(g,h)U(g)U(h) \, , 
\end{equation}
with $\omega(g,h)\in\torus$ ($\T$ denoting the group of unimodular complex numbers).
The mapping $\omega: G \times G \to \torus$, defined through \eqref{eq:projective-rep}, is called the \emph{multiplier} of $U$ and it is required to satisfy
\begin{align*}
& \omega(g,e)=\omega(e,g)=1 \quad \forall g\in G \\
& \omega(g,hk)\omega(h,k)=\omega(g,h)\omega(gh,k) \quad \forall g,h,k\in G \, .
\end{align*}
Clearly, a unitary representation of $G$ is a special type of a projective unitary representation; in that case $\omega(g,h)=1$ for all $g,h\in G$.

From the quantum theoretic point of view, only transformation of \emph{rays} of vectors is relevant.
Therefore, if we multiply unitary operators $U(g)$ with numbers $f(g)\in\torus$, we get a new projective representation $g\mapsto f(g)U(g)$ but the transformation of rays has not been changed.
Two projective unitary representations $U$ and $U'$ are called \emph{similar} if there exists a function $f:G\to\torus$ such that $U'(g)=f(g)U(g)$ for all $g\in G$.
A projective unitary representation $U$ is similar to a unitary representation if and only if its multiplier $\omega$ is \emph{exact}, i.e., there exists a function $f:G\to\torus$ such that
\begin{equation}\label{eq:exact}
\omega(g,h)=f(g)f(h)\overline{f(gh)} \qquad \forall g,h\in G \, .
\end{equation}
Indeed, this is equivelent to $U$ being similar to the unitary representation $U'$ given by $U'(g) = f(g)U(g)$

\begin{example}\label{ex:cyclic}(\emph{Cyclic groups})
Every projective unitary representation $U$ of the cyclic group $\Z_d$ is similar to a unitary representation.
Namely, a repeated application of \eqref{eq:projective-rep} shows that, for each $k\in\Z$, there exists $\alpha(k) \in\torus$ such that
\begin{equation}
U(1)^k = \alpha(k) U(k\, ({\rm mod}\, d)) \, .
\end{equation}
Let $\alpha (d) = \e^{i\theta}$, $\theta\in\R$, and, for each $k\in\Z$, define
\begin{equation}\label{eq:ck}
U^\prime(k) := \left(\e^{-\frac{i\theta}{d}} U(1)\right)^k = \e^{-\frac{i k\theta}{d}} \alpha(k) U(k\, ({\rm mod}\, d)) \, .
\end{equation}
Clearly, $U^\prime$ is a unitary representation of $\Z$. Moreover, $U^\prime$ factors to a unitary representation of $\Z_d$ as
$$
U^\prime(d) = \e^{-i\theta} \alpha(d) \id = \id \, .
$$
Finally, $U$ and $U'$ are similar by \eqref{eq:ck}.
\end{example}

In our later investigations the following criterion will be useful.
(Recall that a subspace $\hik\subset\hi$ is called $U$-invariant if $U(g) v \in\hik$ for all $g\in G$ and $v\in\hik$.)

\begin{proposition}\label{prop:exact}
If a projective unitary representation $U$ has an invariant $1$-dimensional subspace, then it is similar to a unitary representation.
\end{proposition}

\begin{proof}
Let $U$ be a projective unitary representation with an invariant $1$-dimensional subspace.
Hence, there exists a nonzero vector $v\in\hi$ such that $U(g) v \in \C v$ for all $g\in G$.
We can thus define a map $f : G\to \T$ such that $U(g) v = f(g) v$. 
It follows that
\begin{align*}
f(gh) v & = U(gh) v = \omega(g,h) U(g) U(h) v \\
& = \omega(g,h) f (g) f (h) v \, .
\end{align*}
Therefore, $\omega(g,h) = f(gh) \overline{f(g)f(h)}$, hence $\omega$ is an exact multiplier.
\end{proof}

A useful trick when having a projective unitary representation $U$ is to pass from $U$ to a related unitary representation $\widetilde{U}$ that acts in the Hilbert space $\lh$ of operators.
For each $g\in G$ and $L\in\lh$, we define
\begin{equation}
\widetilde{U}(g) L := U(g) L U(g)^\ast \, .
\end{equation}
Notice that even if $U$ is a projective unitary representation, $\widetilde{U}$ is an ordinary unitary representation, since
\begin{align*}
\widetilde{U}(gh) L &= \omega(g,h) \overline{\omega(g,h)} U(g)U(h) L U(h)^\ast U(g)^\ast \\
& = \widetilde{U}(g) \widetilde{U}(h) L \, .
\end{align*}
Moreover, it is easy to see that $U(g)$ commutes with an operator $L\in\lh$ if and only if $\widetilde{U}(g) L =L$.
Since every operator $U(g)$ commutes with the identity operator $\id$, we conclude that the $1$-dimensional subspace $\C\id$ is invariant under $\widetilde{U}$.
In particular, $\widetilde{U}$ is a reducible representation whenever $\dim\hi\geq 2$.

%%%%%%%%%%%%%%%%%%%%%%%
\subsection{Structure of covariant observables}
%%%%%%%%%%%%%%%%%%%%%%%

We will next recall some basic facts about covariant observables.
More details and further references can be found e.g. in \cite{SSQT01}.

Let $H$ be a (proper) subgroup of the group $G$. In the following, we will choose $\Omega\equiv G/H$, i.e., our outcome space will be the quotient space consisting of left cosets $gH=\{gh:h\in H\}$, $g\in G$.
We do not assume that $H$ is a normal subgroup, therefore $G/H$ need not be a group.
However, there is a natural action of $G$ on $G/H$; for each $g'\in G$ and $gH\in \Omega$, we denote by $g'\cdot gH$ the left coset $g'gH$.
This action is transitive, meaning that for any two points $g_1H,g_2H\in\Omega$, there exists $g\in G$ such that $g(g_1H)=g_2H$.

Let $U$ be a projective unitary representation of $G$.
Then $G$ acts, on the one hand, on the outcome space $\Omega$, and on the other hand, on the space $\lh$ through the unitary representation $\widetilde{U}$.
A covariant observable has the property of intertwining these two actions.

\begin{definition}
An observable $\Mo$ based on $\Omega$ is $U$-\emph{covariant} if 
\begin{equation}\label{eq:covariance}
\widetilde{U}(g) \Mo(x) = \Mo(g\cdot x)
\end{equation}
for all $g\in G$, $x\in \Omega$.
\end{definition}

It follows from \eqref{eq:covariance} that
\begin{equation}\label{eq:SM-invariant}
\widetilde{U}(g) \S_\Mo= \S_\Mo
\end{equation}
for all $g \in G$, meaning that $\S_\Mo$ is a $\widetilde{U}$-invariant subspace.
This implies that also $\S_\Mo^\perp$ is a $\widetilde{U}$-invariant subspace.
The unitary representation $\widetilde{U}$ therefore splits into the direct sum 
\begin{equation}
\widetilde{U}= \widetilde{U}^{\S_\Mo} \oplus \widetilde{U}^{\S_\Mo^\perp} \, , 
\end{equation}
where $\widetilde{U}^{\S_\Mo}$ and $\widetilde{U}^{\S_\Mo^\perp}$ are the restrictions of $\widetilde{U}$ to $\S_\Mo$ and $\S_\Mo^\perp$, respectively.

We also see that a $U$-covariant observable $\Mo$ must be of the form
\begin{equation}\label{eq:covariant-povm}
\Mo(gH)=\widetilde{U}(g) M 
\end{equation}
for some positive operator $M\in\lh$.
Namely, we denote $M\equiv \Mo(eH)$ for the identity element $e\in G$ and then \eqref{eq:covariant-povm} follows from \eqref{eq:covariance}.
Notice that if $g',g\in G$ are such that $g'\in gH$, then $\Mo(g'H)=\Mo(gH)$, implying that
\begin{equation}
U(g) M U(g)^\ast = U(g') M U(g')^\ast \, .
\end{equation}
This means that $[M,U(h)]=0$ for all $h\in H$.

If we start from a positive operator $M\in\lh$ satisfying $[M,U(h)]=0$ for all $h\in H$ and define $\Mo$ by formula \eqref{eq:covariant-povm}, then the positivity and the covariance condition \eqref{eq:covariance} holds but we need to check that $\sum_x \Mo(x)=\id$ in order to get an observable. 
This normalization is not automatically satisfied, and the suitable operators $M$ depend on the projective representation $U$.

\begin{example}(\emph{Irreducible projective representation})
Suppose that a projective unitary representation $U$ of $G$ is irreducible. 
Fix a nonzero positive operator $M\in\lh$ and define $\Mo$ on $G$ as $\Mo(g)=U(g)MU(g)^\ast$.
In this example $H=\{e\}$, hence the commutativity condition $[M,U(h)]=0$ for all $h\in H$ puts no constrains on $M$.
For every $g'\in G$, we have
\begin{align*}
U(g') \left( \sum_g \Mo(g) \right) & = \sum_g U(g')U(g) M U(g)^\ast \\
= & \sum_g \overline{\omega(g',g)} U(g'g) M U(g)^\ast U(g')^\ast U(g') \\
 =& \left( \sum_g \Mo(g) \right) U(g') \, .
\end{align*}
Therefore, $\sum_g \Mo(g)=c \id$ for some $c\in\real$, and we see that 
\begin{equation*}
c=\tr{\sum_g \Mo(g)}/\tr{\id}=\# G \cdot \tr{M} / d \, .
\end{equation*}
Redefining $M\to \frac{1}{c} M$ we thus obtain a covariant observable.
\end{example}

%%%%%%%%%%%%%%%%%%%%%%%%%%%%%%%%
\subsection{Existence of covariant observables}
%%%%%%%%%%%%%%%%%%%%%%%%%%%%%%%%

Let $U$ be a projective unitary representation of $G$.
We now pose our earlier question in a modified form:
\begin{quote}
If a linear subspace $\S\subseteq\lh$ is given, is it possible to find a \emph{$U$-covariant} observable $\Mo$ such that $\S=\S_\Mo$?
\end{quote}
Obviously, $\S$ must satisfy the basic criterion \eqref{eq:operator-system} and the $\widetilde{U}$-invariance requirement \eqref{eq:SM-invariant}, i.e., 
\begin{equation}
\id\in \S \, , \qquad \S=\S^\ast \, , \qquad \widetilde{U}(g) \S= \S \quad  \forall g\in G \, .
\end{equation}
But it will turn out that these conditions are not sufficient.

As we have seen earlier, for every $U$-covariant observable $\Mo$, there is a positive operator $M\in\lh$ such that \eqref{eq:covariant-povm} holds. 
It follows that
\begin{equation}
\S_\Mo= \spanno{ \widetilde{U}(g) M : g \in G } \, .
\end{equation}
Therefore, the answer to the existence question can be affirmative only if there exists a positive operator $M\in \S$ such that
\begin{equation}
\S=\spanno{ \widetilde{U}(g) M : g \in G }  \, .
\end{equation}
In mathematical terms, this condition  means that $M$ is a \emph{cyclic vector} for the subrepresentation $\widetilde{U}^\S$, the restriction of $\widetilde{U}$ on the invariant subspace $\S$. Before we present a result that gives an important necessary condition to the existence question, let us recall the concepts of Schmidt rank and cyclic representation.

Let $\hi_1,\hi_2$ two Hilbert spaces.
A vector $v\in\hi_1\otimes\hi_2$ can be written in the so-called \emph{Schmidt form} 
\begin{equation}
v=\sum_j \sqrt{\lambda_j} v^{(1)}_j\otimes v^{(2)}_j \, , 
\end{equation}
where $\{v^{(1)}_j\}\subset\hi_1$ and $\{v^{(2)}_j\}\subset\hi_2$ are orthogonal sets.
The number of nonzero coefficients $\lambda_j$ is called the \emph{Schmidt rank} of $v$ and we denote it by $\rank{v}$.
Let us notice that the Schmidt rank of $v$ is the same as the rank of $v$ regarded as a linear operator from $\hi_1^\ast$ to $\hi_2$.

Let $V$ be a representation of $G$ acting on the Hilbert space $\hi$. 
A vector $v\in\hi$ is \emph{cyclic} for $V$ if the linear span of the set $\{ V(g)v : g\in G\}$ is $\hi$.
If there exists a cyclic vector, then we say that $V$ is a \emph{cyclic representation}.

The following result is a particular case of \cite[Theorem 1.10]{GM71}, the only difference being that we provide an explicit way to construct cyclic vectors.

\begin{proposition}\label{prop:cyclic}
Let $\widehat{G}$ be the (finite) set of irreducible unitary representations of $G$, each $\pi\in\widehat{G}$ acting in the Hilbert space $\kk_\pi$. 
Then we have the following facts.

(1) Let $\pi\in\widehat{G}$, and suppose $v\in\kk_\pi\otimes\vv_\pi$. 
Then the vector $v$ is cyclic for the representation $\pi\otimes\id_{\vv_\pi}$ if and only if $\rank{v} = \dim \vv_\pi$.

(2) Suppose $V$ is a unitary representation of $G$ in $\hh$, and let
$$
\hh = \oplus_{\pi\in\widehat{G}} \, \kk_\pi\otimes\vv_\pi \qquad V = \oplus_{\pi\in\widehat{G}} \, \pi\otimes \id_{\vv_\pi}
$$
be the isotypic decomposition of $V$. 
Here, $\dim\vv_\pi$ is the multiplicity of each irreducible unitary representation $\pi$ in $V$. 
For each $\pi\in\widehat{G}$, let $P_\pi : \hh \to \kk_\pi\otimes\vv_\pi$ be the projection of $\hh$ onto $\kk_\pi\otimes\vv_\pi$. Then, a vector $v\in\hh$ is cyclic for $V$ if and only if $P_\pi v$ is cyclic for the representation $\pi\otimes \id_{\vv_\pi}$ for all $\pi\in\widehat{G}$.

(3) With the notations of item (2), $V$ is a cyclic representation if and only if $\dim\vv_\pi \leq \dim\kk_\pi$ for all $\pi\in\widehat{G}$.
\end{proposition}

\begin{proof}
(1) For a fixed a linear basis $\{k_1,k_2,\ldots,k_d\}$ of $\kk_\pi$, with dual basis $\{k^\ast_1,k^\ast_2,\ldots,k^\ast_d\}$, there exist vectors $\{v_1,v_2,\ldots,v_d\}$ in $\vv_\pi$ such that
$$
v=\sum_i k_i\otimes v_i \, .
$$
Clearly, the dimension of the linear space $\vv_\pi^0 = \spanno{v_1,v_2,\ldots,v_d}$ is the rank of $v$. 

By irreducibility of the representation $\pi$, the algebra $\aa = \spanno{\pi(g) : g\in G}$ coincides with the whole $\elle{\kk_\pi}$. 
Indeed, its commutant $\aa'=\C \id_{\kk_\pi}$ by Schur lemma, hence $\aa = \aa'' = \elle{\kk_\pi}$.
For all $k\in\kk_\pi$, the operator $A=kk^\ast_i$ thus belongs to $\aa$, and $(A\otimes\id_{\vv_\pi}) v = k\otimes v_i$. It follows that
$$
\spanno{(\pi(g)\otimes\id_{\vv_\pi}) v : g\in G} = (\aa\otimes\id_{\vv_\pi}) v \supseteq \kk_\pi\otimes\vv^0_\pi \, .
$$
On the other hand, the reverse inclusion is trivial, hence the equality
$$
\spanno{(\pi(g)\otimes\id_{\vv_\pi}) v : g\in G} = \kk_\pi\otimes\vv_\pi
$$
holds if and only if $\vv^0_\pi = \vv_\pi$, i.e.,~$\rank{v} = \dim\vv_\pi$.

(2) Each map $P_\pi : \hh \to \kk_\pi\otimes\vv_\pi$ satisfies $P_\pi V(g) = (\pi(g)\otimes\id_{\vv_\pi}) P_\pi$ for all $g$. If $v\in\hh$ is cyclic for $V$, then necessarily $P_\pi v$ is cyclic for $\pi\otimes\id_{\vv_\pi}$, as the condition
\begin{align*}
0 & = \ip{(\pi(g)\otimes\id_{\vv_\pi}) P_\pi v}{P_\pi w} \\
& = \ip{V(g) v}{P_\pi w} \qquad \forall g\in G
\end{align*}
implies $P_\pi w = 0$. Conversely, suppose that $P_\pi v$ is cyclic for $\pi\otimes\id_{\vv_\pi}$ for all $\pi\in\widehat{G}$, and let $w\in\hh$ be such that $\ip{V(g) v}{w} = 0$ for all $g\in G$. By \cite[Th\'eor\`eme 8]{Serre},
$$
P_\pi = \frac{\dim\kk_\pi}{\#G} \sum_g \tr{\pi(g)} V(g) \, ,
$$
which implies
\begin{align*}
& \ip{(\pi(g)\otimes\id_{\vv_\pi}) P_\pi v}{P_\pi w} =
\ip{V(g) v}{P_\pi w} \\
& \qquad \qquad \qquad = \frac{\dim\kk_\pi}{\#G} \sum_h \tr{\pi(h)} \ip{V(g) v}{V(h) w} \\
& \qquad \qquad \qquad = \frac{\dim\kk_\pi}{\#G} \sum_h \tr{\pi(h)} \ip{V(h^{-1} g) v}{w} \\
& \qquad \qquad \qquad = 0 \, .
\end{align*}
By cyclicity of $P_\pi v$ in $\kk_\pi\otimes\vv_\pi$, then it follows $P_\pi w = 0$. Since this holds for all $\pi\in\widehat{G}$, we have $w=0$. Hence, $v$ is cyclic for $V$.

(3) If $v\in\hh$ is a cyclic vector for $V$, then $P_\pi v$ must be cyclic for $\pi\otimes\id_{\vv_\pi}$ for all $\pi\in\widehat{G}$ by item (2). By item (1), this implies $\rank{P_\pi v} = \dim\vv_\pi$, which can happen only if $\dim\vv_\pi\leq\dim\kk_\pi$.
\end{proof}

Let $\mathcal{F}(G)$ be the vector space of all complex valued functions on $G$, with the inner product $\ip{f_1}{f_2} = \sum_g \overline{f_1(g)} f_2(g)$.
We recall that the \emph{regular representation} $R$ of $G$ acts on $\mathcal{F}(G)$ and is defined as 
\begin{equation}
[R(g)\phi](g')=\phi(g^{-1}g') \, .
\end{equation}
The regular representation is reducible and each irreducible unitary representation $\pi\in\widehat{G}$ occurs in $\mathcal{F}(G)$ with a multiplicity equal to its dimension.
Therefore, from Proposition \ref{prop:cyclic} we conclude the following consequence.

\begin{corollary}\label{prop:regular}
Let $V$ be a unitary representation of $G$.
Then $V$ is a cyclic representation if and only if it is equivalent to a subrepresentation of the regular representation $R$ of $G$.
\end{corollary}

\begin{example}\label{ex:1dim}(\emph{Representation consisting of $1$-dimensional irreps.})
Suppose that $V$ is a direct sum of $1$-dimensional representations of $G$.
Then $V$ is cyclic if and only if every $1$-dimensional representation $\chi\in\widehat{G}$ is contained in $V$ at most once.
Indeed, the regular representation contains each $1$-dimensional representation exactly once.
The claim thus follows from Corollary \ref{prop:regular}.
\end{example}

%%%%%%%%%%%%%%%%%%%%%%%%%%%%%%
\section{Existence of covariant PIC observables}\label{sec:existence}
%%%%%%%%%%%%%%%%%%%%%%%%%%%%%%

It has now become clear that the existence of a covariant observable that is pure-state informationally complete depends crucially on the group $G$ and its projective unitary representation  $U$.
In this section we derive some conditions that preclude the existence of a $U$-covariant PIC observable.

\begin{proposition}\label{prop:abelian1}
Let $U$ be a projective unitary representation of the group $G$ on $\hh$. 
Suppose there exist two linearly independent vectors $v_1 , v_2 \in \hh$ and functions $f_1 , f_2 : G\frecc \T$ such that $U(g) v_i = f_i (g) v_i$ for all $g\in G$ and $i\in\{1,2\}$. Then there exists no $U$-covariant PIC observable.
\end{proposition}

\begin{proof}
We define two pure states $\varrho_i = \kb{v_i}{v_i} / \no{v_i}^2$, $i=1,2$ and we will show that no $U$-covariant observable can separate these states. Notice that $\varrho_1 \neq \varrho_2$ by the linear independence of $\{v_1,v_2\}$. 

Suppose $\Mo$ is a $U$-covariant observable based on a quotient space $\Omega=G/H$ and let $i\in\{1,2\}$. 
Then, for every $x=gH\in\Omega$,
\begin{align*}
\tr{\Mo(gH)\varrho_i} & = \tr{U(g)\Mo(eH)U(g)^\ast \varrho_i} \\
& = f_i (g) \overline{f_i (g)} \tr{M\varrho_i} = \tr{M\varrho_i} \, ,
\end{align*}
hence the map $x\mapsto\tr{\Mo(x)\varrho_i}$ is constant. 
Since 
\begin{equation}
\sum_x \tr{\Mo(x)\varrho_i} = \tr{\id \varrho_i} = 1 \, , 
\end{equation}
we must have
\begin{equation}
\tr{\Mo(x)\varrho_i} = \frac{1}{\# \Omega} \qquad \forall x\in \Omega \, .
\end{equation}
In particular, $\tr{\Mo(x)\varrho_1}=\tr{\Mo(x)\varrho_2}$ for all $x\in \Omega$.
Therefore, $\Mo$ is not PIC. 
\end{proof}

\begin{proposition}\label{prop:abelian2}
Let $G$ be an abelian group and $\dim\hh\geq 2$. 
Then, in the following two cases there exists no $U$-covariant PIC observable:
\begin{itemize}
\item[(a)] if $U$ is a unitary representation;
\item[(b)] if $G$ is cyclic and $U$ is a projective unitary representation.
\end{itemize}
\end{proposition}

\begin{proof}
(a) If $U$ is a unitary representation, then the Hilbert space $\hh$ decomposes into the orthogonal sum of $1$-dimensional subrepresentations, each one carrying the action of a $1$-dimensional unitary representation of $G$. Then the claim follows by Proposition \ref{prop:abelian1}. 
\newline\indent
(b) If $G$ is cyclic and $U$ is a projective unitary representation, then $U$ is similar to a unitary representation (see Example \ref{ex:cyclic}), and the claim follows from (a).
\end{proof}

%%%%%%%%%%%%%%%%%%%%%%%%%%%%%%
\section{Existence of covariant and minimal PIC observables}\label{sec:minimal}
%%%%%%%%%%%%%%%%%%%%%%%%%%%%%%

In this section we combine the earlier concepts and methods and search for a minimal PIC observable that is covariant under some projective unitary representation of some finite group.
As it turns out, the existence of such an observables depends on the Hilbert space dimension $d$.
In particular, we will investigate the existence question for the dimensions $d=3$ and $d=7$. 
These two instances demonstrate that a desired observable may exist or not.

%%%%%%%%%%%%%%%%%%%%%%%%%%%
\subsection{Dimension 3}\label{sec:dim3}
%%%%%%%%%%%%%%%%%%%%%%%%%%%

We first investigate the case when the dimension of $\hi$ is $3$. 
We want to find an observable $\Mo$ such that $\Mo$ is PIC and has minimal number of outcomes, and further that $\Mo$ is covariant with respect to some group $G$.
We will restrict our search for observables based on $\Omega\equiv G$ since already in this situation we can find two possible symmetry groups.

The minimal number of outcomes for a PIC in dimension $3$ is $8$ \cite{HeMaWo11}.
There are five groups with $8$ elements: three abelian groups $\Z_8$, $\Z_2\times \Z_4$, $\Z_2\times \Z_2 \times \Z_2$ and two non-abelian groups $D$ (dihedral group) and $Q$ (quaternionic units).
We will show that the abelian groups are unsuitable while the non-abelian groups can be used to construct a desired observable.

Let $G$ be a group with $8$ elements, $U$ a projective unitary representation of $G$ and $\Mo$ a $U$-covariant PIC observable on $G$. 
In the following we will proceed in steps to reveal the limitations in the choice of $G$ and $U$.

%%%%%%%%%%%%%%%%%%%%%%%%%%%%%%%%%%%%%%%%%%
\subsection*{$\bullet$ Projective representation $U$ must be reducible}
%%%%%%%%%%%%%%%%%%%%%%%%%%%%%%%%%%%%%%%%%%

For a minimal PIC observable $\Mo$ we have $\dim \S_\Mo=8$, hence the orthogonal space $\S_\Mo^\perp$ is generated by a single selfadjoint operator $T$ which satisfies $\tr{T}=0$ and by Proposition \ref{prop:pic} must have rank $3$, i.e., is invertible.
Thus, from \eqref{eq:SM-invariant} it follows that
\begin{equation}
U(g) T U(g)^\ast = c_g T \qquad \forall g \in G
\end{equation}
for some numbers $c_g\in\R$.
This implies that 
\begin{equation*}
\det (T)=\det(U(g) T U(g^\ast))=\det (c_g T) = c_g^3 \det(T) \, .
\end{equation*}
Since $\det (T)\neq 0$, we conclude that $c_g=1$ for every $g\in G$.
Therefore,
\begin{equation}
U(g) T=TU(g) \qquad \forall g \in G \, .
\end{equation}
But since $\tr{T}=0$, we see that $T$ cannot be a scalar multiple of the identity operator.
Therefore, the projective representation $U$ is reducible.

%%%%%%%%%%%%%%%%%%%%%%%%%%%%%%%%%%%%%%%%%%
\subsection*{$\bullet$ Projective representation $U$ is similar to an ordinary unitary representation}
%%%%%%%%%%%%%%%%%%%%%%%%%%%%%%%%%%%%%%%%%%

We can infer more about $U$ by using the spectral decomposition for $T$.
By changing $T$ to $-T$ if necessary (this does not change $\S_\Mo^\perp$), we can write $T$ as
\begin{equation}
T = \lambda_1 P_1 - \lambda_2 P_2 - \lambda_3 P_3 \, , 
\end{equation}
where $\lambda_j >0$, $\lambda_1=\lambda_2+\lambda_3$ and $P_1,P_3,P_3$ are orthogonal $1$-dimensional projections.
Depending on the eigenvalues of $T$, we have two alternative situations:
\begin{itemize}
\item[(a)] If $\lambda_2\neq \lambda_3$, then $U$ commutes with each of the three projections $P_1$, $P_2$ and $P_3$.
\item[(b)] If $\lambda_2 = \lambda_3$, then $U$ commutes with $P_1$ and $P_2+P_3$ (but not necessarily with $P_2$ and $P_3$ separately).
\end{itemize}
In both cases $U$ leaves invariant the $1$-dimensional subspace $P_1 \hh$.
By Proposition \ref{prop:exact} we conclude that $U$ is similar to an ordinary unitary representation.

%%%%%%%%%%%%%%%%%%%%%%%%%%%%%%%%%%%%%%%%
\subsection*{$\bullet$ Symmetry group must be non-abelian}
%%%%%%%%%%%%%%%%%%%%%%%%%%%%%%%%%%%%%%%%%

We have seen that $U$ must be an ordinary unitary representation of $G$.
It follows from Proposition \ref{prop:abelian2} that the symmetry group $G$ must be non-abelian, hence either the dihedral group $D$ or the group of unit quaternions $Q$.

The two situations that were separated according to the eigenvalues of $T$ lead to different conclusions:
\begin{itemize}
\item[(a)]
If $\lambda_2\neq \lambda_3$, then $U$ leaves invariant all the $1$-dimensional spaces $P_1\hh$, $P_2\hh$ and $P_3\hh$. By Proposition \ref{prop:abelian1}, there exists no $U$-covariant PIC observable.

\item[(b)]
If $\lambda_2=\lambda_3$, then $U$ leaves invariant the $1$-dimensional space $P_1\hh$ and the $2$-dimensional space $(P_2+P_3)\hh$. If the space $(P_2+P_3)\hh$ is irreducible, then Proposition \ref{prop:abelian1} does not exclude the existence of $U$-covariant PIC observables. To determine whether $U$-covariant PIC observables exist or not, one further needs to establish if the representation $\widetilde{U}$ is cyclic.
\end{itemize}

%%%%%%%%%%%%%%%%%%%%%%%%%%%%%%%%%%%%%%%%%
\subsection*{$\bullet$ Quaternionic and dihedral symmetry groups}
%%%%%%%%%%%%%%%%%%%%%%%%%%%%%%%%%%%%%%%%%

The quaternionic group $Q$ consists of $8$ elements $\pm 1, \pm i , \pm j , \pm k$ satisfying the relations
\begin{equation}
\begin{split}
(-1)^2 & = 1 \qquad (\pm 1) g = g(\pm 1)= \pm g \quad \forall g\in Q \\
i^2 & = j^2 = k^2 = -1 \qquad ij = -ji = k \, .
\end{split}
\end{equation}
For our purposes, it is convenient to use a matrix realization of $Q$.
We denote by $M_2 (\C)$ the Hilbert space of complex $2\times 2$ matrices, equipped with the Hilbert-Schmidt inner product.
The identity matrix $\id$ together with the Pauli matrices $\sigma_1$, $\sigma_2$, $\sigma_3$ form an orthogonal basis of $M_2(\C)$.
The quaternionic group $Q$ can be described as the collection of matrices $\{\pm\id \, ,\, \pm\sigma_1 \, ,\, \pm\sigma_2 \, ,\, \pm\sigma_3\}$, endowed with the usual matrix product rule, according to the correspondence
\begin{align*}
\pm 1 \leftrightarrow \pm \id \, , \quad \pm i \leftrightarrow \pm i \sigma_1 \, , \quad \pm j \leftrightarrow \mp i \sigma_2 \, , \quad  \pm k \leftrightarrow \pm i \sigma_3 \, .
\end{align*}
 
The dual $\widehat{Q}$ consists of four $1$-dimensional unitary representations and a single $2$-dimensional unitary representation. 
The $2$-dimensional unitary representation is the identity map $\pi (g) = g$, and the $1$-dimensional representations are given in Table \ref{tab:q}.

\begin{table}
\caption{The $1$-dimensional irreducible representations of the quaternionic group $Q$ and the dihedral group $D$.}
\label{tab:q}
\begin{center}
\begin{tabular}{ | c || l | l |}
\hline
irrep & value in the case of $Q$ & value in the case of $D$  \\
\hline\hline
$\chi_0(g)$  & 1 \quad $\forall g$ & 1 \quad $\forall g$ \\ 
\hline
$\chi_1(g)$ & 1 if $g\in\{\pm\id \, , \, \pm i\sigma_1\}$ & 1 if $g\in\{\pm\id \, , \, \pm i\sigma_1\}$  \\
 & -1 otherwise  &  -1 otherwise \\ 
\hline
\hline
$\chi_2(g)$ & 1 if $g\in\{\pm\id \, , \, \pm i\sigma_2\}$ & 1 if $g\in\{\pm\id \, , \, \pm \sigma_2\}$  \\
 & -1 otherwise  &  -1 otherwise \\ 
\hline
\hline
$\chi_3(g)$ & 1 if $g\in\{\pm\id \, , \, \pm i\sigma_3\}$ & 1 if $g\in\{\pm\id \, , \, \pm \sigma_3\}$  \\
 & -1 otherwise  &  -1 otherwise \\ 
\hline
\end{tabular}
\end{center}
\end{table}

The unitary representation
$$
\widetilde{\pi}(g) L := \pi(g) L \pi(g)^\ast \qquad \forall L\in M_2 (\C)
$$
decomposes into the direct sum
$$
M_2 (\C) = \C \id \oplus \C \sigma_1 \oplus \C \sigma_2 \oplus \C \sigma_3 \, , \qquad \widetilde{\pi} = \chi_0 \oplus \chi_1 \oplus \chi_2 \oplus \chi_3 \, .
$$
The contragradient representation $\bar{\pi}$ in the dual space $\C^{2\,\ast}$ of row vectors is defined as
\begin{equation}
\bar{\pi} (g) v
^t : = v
^t \pi(g^{-1}) \qquad \forall v
\in\C^2 \, . 
\end{equation}
The representation $\bar{\pi}$ is equivalent to $\pi$, and an intertwining operator $V:\C^{2\,\ast}\to\C^2$ is given by
\begin{equation}\label{eq:defV}
V v
^t = \sigma_2 v
  \qquad \forall v
\in \C^2 \, .
\end{equation}
In the following it is convenient to use block form for matrices, for instance
\begin{equation*}
\mat{0 & \bar{v
}^t \\ v
 & 0} \equiv \mattt{0 & \overline{v}_1 & \overline{v}_2 \\
 v_1 & 0 & 0 \\
 v_2 & 0 & 0 }
\end{equation*}
where $v_1,v_2$ are the components of the vector $v\in\C^2$.

We now consider the Hilbert space $\hh = \C^3$ and introduce the following unitary representation $U$ in $\hh$:
\begin{equation}\label{eq:U}
U(g) = \mat{1 & 0 \\ 0 & \pi(g)} \, .
\end{equation}
If $\lambda >0$ and $T$ is defined as
\begin{equation}\label{eq:T}
T = \mat{2\lambda & 0 \\ 0 & -\lambda \id_{\C^2} } \, ,
\end{equation}
then clearly $\tr{T}=0$ and $U(g) T U(g)^\ast = T$ for all $g$. 
These choices satisfy the necessary requirements found earlier; $U$ is reducible and $T$ has a degenerate eiegenvalue. 

With these preliminary observations, we are ready for the following result.

\begin{proposition}\label{prop:covPIC3}
Let $M\in\lh$ be the operator
\begin{align}\label{eq:seed}
M &= \frac{1}{8} \id + \alpha_1 \mat{0 & 0 \\ 0 & \sigma_1} + \alpha_2 \mat{0 & 0 \\ 0 & \sigma_2} + \alpha_3 \mat{0 & 0 \\ 0 & \sigma_3} \nonumber \\
&+ \mat{0 & \bar{v
}^t \\ v
 & 0} 
\end{align}
with
\begin{gather*}
\alpha_1 \, , \, \alpha_2 \, , \, \alpha_3 \in \R \, , \qquad  v
\in\C^2 
\end{gather*}
such that
\begin{align}
& \alpha_1 \neq 0 \, , \quad \alpha_2 \neq 0 \, , \quad \alpha_3 \neq 0 \label{cond:1} \\
& v \neq 0 \label{cond:2}
\end{align}
and
\begin{gather}
M\geq 0 \, . \label{cond:3}
\end{gather}

Then the map
\begin{equation}\label{eq:POVM}
\Mo(g) = U(g) M U(g)^\ast \, , \qquad g\in Q
\end{equation}
is a $U$-covariant observable.
Moreover, $\S_\Mo = T^\perp$, hence $\Mo$ is a PIC with minimal number of outcomes.
\end{proposition}

\begin{proof}
We first show that $\S_\Mo = T^\perp$, i.e.,~the operator $M$ is a cyclic vector  for the restriction of the representation $\widetilde{U}$ to the invariant subspace $T^\perp$ of $\lh$. 
Under the action of $\widetilde{U}$, the space $T^\perp$ decomposes into the direct sum of irreducible invariant subspaces
\begin{eqnarray*}
T^\perp & = & \C \id \oplus \C \mat{0 & 0 \\ 0 & \sigma_1} \oplus \C \mat{0 & 0 \\ 0 & \sigma_2} \oplus \C \mat{0 & 0 \\ 0 & \sigma_3} \\
&& \oplus \mat{0 & \C^{\ast 2} \\ 0 & 0} \oplus \mat{0 & 0 \\ \C^2 & 0} \, ,
\end{eqnarray*}
and according to such splitting the representation $\widetilde{U}$ restricted to $T^\perp$ decomposes as
\begin{eqnarray*}
\widetilde{U} & = & \chi_0 \oplus \chi_1 \oplus \chi_2 \oplus \chi_3 \oplus \bar{\pi} \oplus \pi \, .
\end{eqnarray*}

As $\pi$ and $\bar{\pi}$ are equivalent, the representations $\bar{\pi} \oplus \pi$ and $\pi\otimes\id_{\C^2}$ are equivalent.
A linear map
$$
W : \mat{0 & \C^{\ast 2} \\ 0 & 0} \oplus \mat{0 & 0 \\ \C^2 & 0} \to \C^2 \otimes \C^2
$$
explicitely yelding this equivalence is given by
$$
W \mat{0 & u^t \\ v
 & 0} = v
 \otimes \left(\begin{array}{c} 1 \\ 0 \end{array}\right) + Vu^t \otimes \left(\begin{array}{c} 0 \\ 1 \end{array}\right) \, ,
$$
where $V$ is the map defined in \eqref{eq:defV}.
The condition \eqref{cond:2} is the same of $V\bar{v
}^t \notin \C v
$ \red{(indeed, $\sigma_2 \bar{v
} \notin \C v$ if and only if $v\neq 0$)}, which in turn is equivalent to $\rank{W \mat{0 & \bar{v}^t \\ v & 0}} = 2$. 
By item (1) of Proposition \ref{prop:cyclic}, this implies that the vector $\mat{0 & \bar{v}^t \\ v & 0}$ is cyclic for the subrepresentation $\bar{\pi}\oplus\pi\simeq \pi\otimes\id_{\C^2}$ of $\widetilde{U}$. Moreover, each $1$-dimensional subrepresentation $\chi_i$ of $\widetilde{U}$ is clearly cyclic. 
Condition \eqref{cond:1} and the cyclicity of $\mat{0 & \bar{v}^t \\ v & 0}$ then imply that $M$ is cyclic in $T^\perp$ by item (2) of Proposition \ref{prop:cyclic}, as claimed.

We still need to show that $\Mo$ is an observable.
First, since $\Mo$ is of the form \eqref{eq:POVM} and $M\geq 0$, it follows that every $\Mo(g)$ is positive.
Second, we need to prove that $\Sigma_g \Mo(g)=\id$. 
We denote
\begin{equation}
A := \sum_g \Mo(g) \, .
\end{equation}
Then, $U(g) A = A U(g)$ for all $g$.
It follows that
\begin{equation}
A = \mat{\beta_1 & 0 \\ 0 & \beta_2 \id_{\C^2} } \qquad \beta_1 , \beta_2 \in \C \, .
\end{equation}
Since $A\in T^\perp$, we have $\tr{AT} = 0$, which implies $\beta_1 = \beta_2$. 
On the other hand,
\begin{equation*}
\tr{A} = \sum_g \tr{\Mo(g)} = 8\cdot \tr{M} = 8\cdot \tr{(1/8) \id} = 3 \, , 
\end{equation*}
therefore $\beta_1 = \beta_2=1$.
This means that $A = \id$.
\end{proof}

\begin{remark}
The nonzero eigenvalues of the matrices
\begin{align}
\alpha_1 \mat{0 & 0 \\ 0 & \sigma_1} + \alpha_2 \mat{0 & 0 \\ 0 & \sigma_2} + \alpha_3 \mat{0 & 0 \\ 0 & \sigma_3} 
\end{align}
and
\begin{align}
 \mat{0 & \bar{v}^t \\ v
 & 0} 
\end{align}
are $\pm \sqrt{\alpha_1^2+\alpha_2^2+\alpha_3^2}$ and $\pm \sqrt{|v_1|^2 + |v_2|^2}$, respectively.
Therefore, the positivity condition \eqref{cond:3} is satisfied if we choose $\alpha_1 , \alpha_2 , \alpha_3$ and $v$ such that
\begin{equation*}
\sqrt{\alpha_1^2 + \alpha_2^2 + \alpha_3^2} + \sqrt{|v_1|^2 + |v_2|^2} \leq \frac{1}{8} \, .
\end{equation*}
However, this inequality is not a necessary condition for \eqref{cond:3} to hold, but only a convenient sufficient condition.
\end{remark}

\begin{example}(\emph{Rank-1 PIC observable}) 
Suppose $\gamma\in [0,2\pi)$ and $\alpha_1,\alpha_2,\alpha_3$ are nonzero real numbers satisfying $\alpha_1^2 + \alpha_2^2 + \alpha_3^2 = 1/64$. Then the operator
$$
M = \left( \begin{array}{ccc}
\frac{1}{8} & \frac{\e^{-i\gamma} \sqrt{1+8\alpha_3}}{8} & \frac{\e^{-i\gamma} (\alpha_1 - i\alpha_2)}{\sqrt{1+8\alpha_3}} \\
\frac{\e^{i\gamma} \sqrt{1+8\alpha_3}}{8} & \frac{1}{8} + \alpha_3 & \alpha_1 - i\alpha_2 \\
\frac{\e^{i\gamma} (\alpha_1 + i\alpha_2)}{\sqrt{1+8\alpha_3}} & \alpha_1 + i\alpha_2 & \frac{1}{8} - \alpha_3
\end{array} \right)
$$
is of the form \eqref{eq:seed} and satisfies conditions \eqref{cond:1} and \eqref{cond:2}. Moreover, a direct calculation yelds $M^2 = (3/8) M$, which shows that $M\geq 0$ and $\rank{M} = 1$. It follows from Proposition \ref{prop:covPIC3} that the map $g\mapsto U(g) M U(g)^\ast$ is a $U$-covariant PIC observable.
\end{example}

Except some minor details, the previous construction can be done by using the dihedral group $D$ instead of the quaternionic group $Q$.
We briefly explain the needed modifications.
It is convenient to describe $D$ as a collection of $2\times 2$ complex matrices,
\begin{gather*}
D= \{\id \, ,\, -\id \, ,\, i\sigma_1 \, ,\, -i\sigma_1 \, ,\, \sigma_2 \, ,\, -\sigma_2 \, , \, \sigma_3 \, ,\, -\sigma_3 \} \, .
\end{gather*}
As in the case of $Q$, the dual $\widehat{D}$ consist of four $1$-dimensional unitary representations and a single $2$-dimensional unitary representation. 
The $2$-dimensional unitary representation is the identity map $\pi (g) = g$ , and the $1$-dimensional representations are given in Table \ref{tab:q}.
The unitary representation
$$
\widetilde{\pi}(g) L := \pi(g) L \pi(g)^\ast \qquad \forall L\in M_2 (\C)
$$
decomposes into the direct sum
$$
M_2 (\C) = \C \id \oplus \C \sigma_1 \oplus \C \sigma_2 \oplus \C \sigma_3 \, , \qquad \widetilde{\pi} = \chi_0 \oplus \chi_1 \oplus \chi_2 \oplus \chi_3 \, .
$$
The contragradient representation $\bar{\pi}$ in the dual space $\C^{2\,\ast}$ of row vectors is defined as
\begin{equation}
\bar{\pi} (g) v^t : = v^t \pi(g^{-1}) \qquad \forall v\in\C^2 \, , 
\end{equation}
and it is equivalent to $\pi$. 
An intertwining operator $V:\C^{2\,\ast}\to\C^2$ is given by
\begin{equation}
V v^t = \sigma_3 v  \qquad \forall v\in \C^2 \, .
\end{equation}
The unitary representation $U$ and the operator $T$ are defined similarly as in \eqref{eq:U} and \eqref{eq:T}, respectively.
Then, Proposition \ref{prop:covPIC3} holds if the condition \eqref{cond:2} is replaced with 
\begin{equation}
|v_1| \neq |v_2| \, .
\end{equation}

%%%%%%%%%%%%%%%%%%%%%%%%%%%%%%
\subsection{Dimension $7$}\label{sec:dim7}
%%%%%%%%%%%%%%%%%%%%%%%%%%%%%%

If the dimension of the Hilbert space $\hh$ is $7$, then, according to Table \ref{tab:min}, the minimal number of outcomes for a PIC observable is $23$, which is a prime number. The next result shows that, if $p$ is a prime number, there are no covariant PIC observables with $p$ outcomes. This rules out the existence of minimal covariant PIC observables in dimension $7$.

There are other cases, in addition to $d = 7$, which do not admit minimal covariant PIC observables for the same reason. The dimensions whose minimal PIC observables are known to have a prime number of outcomes (and which consequently do not have minimal covariant PIC observables) are listed in Table \ref{tab:prime} up to $d\leq 1000$. This list is calculated using the results from \cite{HeMaWo11}.

\begin{table}
\caption{The dimensions $d\leq 1000$ that are known to have the property that the number of outcomes of a minimal PIC is a prime number.}
\label{tab:prime}
\begin{center}
\begin{tabular}{ | c | c |}
\hline
dimension & min $\#$ of outcomes  \\
\hline
\hline
7 & 23\\
13 & 47\\
19 & 71\\
21 & 79\\
49 & 191\\ 
67 & 263\\ 
69 & 271\\ 
97 & 383\\ 
259 & 1031\\
261 & 1039\\
273 & 1087\\
289 & 1151\\
321 & 1279\\
517 & 2063\\
529 & 2111\\
\hline
\end{tabular}
\end{center}
\end{table}

\begin{proposition}\label{prop:homogeneus}
Let $G$ be a group, $H\subset G$ a proper subgroup and $U$ a projective unitary representation of $G$ in $\hh$. 
Suppose $\# (G/ H)$ is prime. 
If $\Mo:G/H\to\lh$ is a $U$-covariant observable, it is not PIC.
\end{proposition}

The proof is a consequence of the following lemma.
Let us first notice that any subgroup $G_0\subseteq G$ acts on the quotient space $G/H$ in the natural way; $x\cdot (gH)=(xg)H$ for all $x\in G_0$ and $g\in G$.
But this action need not be transitive. 

\begin{lemma}\label{lemma}
Suppose $G$ is a group and $H\subseteq G$ a subgroup such that $\# (G/ H)$ is prime. 
There exists a cyclic subgroup $G_0\subset G$ such that the action of $G_0$ on $G/H$ is transitive.
\end{lemma}

\begin{proof}
We recall that, if $p$ denotes a prime number, a {\em $p$-Sylow subgroup} $G_p \subseteq G$ is a subgroup whose order is $p^q$ for some integer $q\geq 1$, and $p^q$ is the highest power dividing the order of $G$. For any prime $p$ dividing the order of $G$, there exists a $p$-Sylow subgroup $G_p$ of $G$ (see e.g.~\cite[Theorem 6.1 of Chapter 1]{Lang} or \cite[Theorem 1.7]{Isaacs2}).

Denote $p=\#(G/ H)$, and let $\# G = p^q m$, with $p$ not dividing $m$. Fix a $p$-Sylow subgroup $G_p\subseteq G$. 
Then, $G_p\nsubseteq H$, as otherwise 
\begin{equation*}
\# G = \# (G/ H) \cdot \# (H/ G_p) \cdot \# G_p = p^{q+1} \cdot \# (H/G_p) 
\end{equation*}
and this contradicts the assumption that $\# G = p^q m$, with $p$ not dividing $m$. 
Thus, we can pick $g_0\in G_p$ such that $g_0\notin H$, and we denote by $G_0$ the cyclic subgroup generated by $g_0$. 
Since $G_0$ is a subgroup of $G_p$, we must have $\# G_0 = p^r$ for some $1\leq r\leq q$. 
On the other hand, $G_0 \neq G_0 \cap H$ since $g_0\notin H$. 
It follows that $\# (G_0 / (G_0 \cap H)) = p^s$, with $s\geq 1$. 

We define a map $\Lambda:G_0 /(G_0 \cap H) \to G/H$ by
$$
\Lambda(g(G_0 \cap H)) = gH \quad \forall g\in G_0 \, .
$$
This map is well defined and is an injection, since, for all $g,g'\in G_0$, $g^{-1} g' \in H$ if and only if $g^{-1}g' \in H\cap G_0$.
Therefore, $\# (G_0 /  (G_0 \cap H)) \leq \# (G/ H)$.
Since $\# (G/ H) = p$ and $\#(G_0 / (G_0 \cap H)) = p^s$ with $s\geq 1$, we conclude that $\# (G_0 /  (G_0 \cap H)) = p$.
It follows that $\Lambda$ is a bijection.
It is easy to verify that for $g,g'\in G_0$, we have
\begin{equation*}
g' \cdot\Lambda(g(G_0 \cap H))=g' g G/H \, .
\end{equation*}
Since $G_0$ acts transitively on $G_0 /(G_0 \cap H)$ and $\Lambda$ is a bijection, we see that this formula defines a transitive action of $G_0$ on $G/H$. 
\end{proof}

\begin{proof}[Proof of Proposition \ref{prop:homogeneus}]
Let $\Mo:G/H\to\lh$ be a $U$-covariant observable.
Choose a cyclic subgroup $G_0 \subset G$ as in Lemma \ref{lemma} and let $U'$ be the restriction of $U$ to $G_0$. 
Then, $\Mo$ is a $U'$-covariant observable and not PIC by Proposition \ref{prop:abelian2}.
\end{proof}

We remark that in Lemma \ref{lemma} the condition that $\#(G/ H)$ is prime is essential. 
Indeed, if e.g.~$G = Q$ and $H=\{\id \, ,\, -\id\}$, then $\# (Q/H) = 4$, but there is no cyclic subgroup $G_0 \subseteq Q$ whose action on $Q/H$ is transitive.
Indeed, by direct inspection, one can check that every cyclic subgroup $G_0 \subseteq Q$ has order $2$ or $4$ and contains $H$, so $\# (G_0 /(G_0 \cap H)) = 1$ or $2 \neq 4$.

%%%%%%%%%%%%%%
\section{Conclusions}\label{sec:conclusions}
%%%%%%%%%%%%%%

The minimal number of outcomes for a pure-state informationally complete (PIC) observable has been recently solved in \cite{HeMaWo11}.
One would hope that it is possible to find a minimal PIC observable with some simple mathematical structure, possibly in some symmetric form.
An obvious try is to use covariance with respect to some finite group since this works so nicely in the case of minimal informationally complete observables.
It would provide a simple and easy way to construct minimal observables with the desired property.

We have seen that the approach of using covariant observables fails in general; for some dimensions there are no minimal PIC observables that would be covariant under any finite group.
We conclude that this kind of symmetry is lacking in the duality of observables and pure states, although it exists between observables and all states.

One can turn into a more general question:
Is it possible to have a minimal PIC observable $\Mo$ with all operators having the same set of eigenvalues?
For such $\Mo$ any pair of operators $\Mo(x)$ and $\Mo(x')$ are unitarily equivalent.
Hence, $\Mo$ bears some symmetry although it need not be covariant under any projective unitary representation.

A particularly appealing observable of this type would consists of rank-1 operators.
It has been shown in \cite{Finkelstein04} that in every finite dimension $d$, there exists an observable consisting of $2d$ operators with rank-1 and identifying all pure states up to a measure zero.
This, however, does not answer to the question of rank-1 PIC observables, which are required to identify all states and therefore must have at least $4d-4-\delta(d)$ outcomes \cite{HeMaWo11}.

%%%%%%%%%%%%%%%%%%%%%%%%%%%%%%%%%%%%%%%
\section*{Acknowledgements}
%%%%%%%%%%%%%%%%%%%%%%%%%%%%%%%%%%%%%%%%

The authors wish to thank Jussi Schultz for his comments on an earlier version of this paper.
T.H.~acknowledges financial support from the Academy of Finland (grant no. 138135). 
A.T.~acknowledges the financial support of the Italian Ministry of Education, University and Research (FIRB project RBFR10COAQ).

%%%%%%%%%
%%%%%%%%%
\end{document}